\documentclass[conference, 10pt, twocolumn]{ieeeconf}

\IEEEoverridecommandlockouts
\usepackage[usenames]{color}
\usepackage{enumerate}
\usepackage{url}
\usepackage{subfigure}
\usepackage{amsfonts,mathrsfs}
\usepackage{amssymb,amsmath}
\usepackage{verbatim}
\usepackage{acronym}
\usepackage{mathtools}
\usepackage{graphicx}

\def\fskip#1{}

\newtheorem{theorem}{Theorem}

\newtheorem{definition}{Definition}
\newtheorem{example}{Example}

\newtheorem{lemma}{Lemma}

\newtheorem{proposition}[theorem]{Proposition}
\newtheorem{remark}{Remark}

\renewcommand{\theenumi}{\arabic{enumi}}

\def\1{{\bf 1}}

\def\R{\mathbb{R}}

\newcommand{\remove}[1]{}

\DeclarePairedDelimiter\ceil{\lceil}{\rceil}

\begin{document}

\title{Complexity of Equilibrium in Diffusion Games on Social Networks*}
\author{\authorblockN{Seyed Rasoul Etesami, Tamer Ba\c{s}ar}
  \authorblockA{Coordinated Science Laboratory, University of Illinois at Urbana-Champaign,  Urbana, IL 61801\\
     Email: (etesami1, basar1)@illinois.edu}
\thanks{*Research supported in part by the ``Cognitive \& Algorithmic Decision Making" project grant through the College of Engineering of the University of Illinois, and in parts 
by AFOSR MURI Grant FA 9550-10-1-0573 and NSF grant CCF 11-11342. An earlier version of this paper, dealing with only a small subset of the class of problems addressed here, was presented at the 2014 American Control Conference (ACC), and is listed as \cite{etesami2014complexity} in the References section.}
}
\maketitle
\begin{abstract}
In this paper, we consider the competitive diffusion game, and study the existence of its pure-strategy Nash equilibrium when defined over general undirected networks. We first determine the set of pure-strategy Nash equilibria for two special but well-known classes of networks, namely the lattice and the hypercube. Characterizing the utility of the players in terms of graphical distances of their initial seed placements to other nodes in the network, we show that in general networks the decision process on the existence of pure-strategy Nash equilibrium is an NP-hard problem. Following this, we provide some necessary conditions for a given profile to be a Nash equilibrium. Furthermore, we study players' utilities in the competitive diffusion game over Erdos-Renyi random graphs and show that as the size of the network grows, the utilities of the players are highly concentrated around their expectation, and are bounded below by some threshold based on the parameters of the network. Finally, we obtain a lower bound for the maximum social welfare of the game with two players, and study sub-modularity of the players' utilities.   
\end{abstract}
\begin{keywords}

Competitive diffusion game, pure-strategy Nash equilibrium, sub-modular function, NP-hardness, Erdos-Renyi graphs, social welfare.

\end{keywords}  

\section{Introduction}
In recent years, there has been a wide range of studies on the role of social networks in various disciplinary areas. In particular, availability of large data from online social networks has drawn the attention of many researchers to model the behavior of agents in a social network using the possible interactions among them \cite{jadbabaie2003coordination,goyal2012competitive,bharathi2007competitive}. One of the widely studied models in social networks is the diffusion model, where the goal is to propagate a certain type of product or behavior in a desired way through the network \cite{goyal2012competitive}, \cite{young2002diffusion}, \cite{acemoglu2011diffusion} and \cite{kempe2003maximizing}. Other than applications in online advertising for companies' products, such a model has applications in epidemics and immunization v.s. virus spreading \cite{khanafer2014infected}, \cite{khanafer2014informationspread}. 
One of the challenges in such models has been obtaining the solution to the best seed placement problem, which has been extensively studied for different processes \cite{yildiz2011discrete}, \cite{fazli2012non}, \cite{ackerman2010combinatorial} and \cite{gionis2013opinion}.

In many of the applications in social networks, it is natural to have more than one party that wants to spread information on his own products. This imposes a sort of competition among the providers who are competing for the same set of resources and their goal is to diffuse information on their own product in a desired way across the society. Such a competition can be modeled within a game theoretic framework \cite{basar1999dynamic}, and hence, a natural question one can ask is characterization of the set of equilibria of such a game. Several papers in the literature have in fact addressed this question in different settings, with some representative ones being \cite{goyal2012competitive}, \cite{ghaderi2013opinion}, \cite{bharathi2007competitive}, \cite{branzei2011social}, \cite{richardson2002mining} and \cite{singer2012win}. Our goal in this paper is to expand on this literature by addressing the issue of complexity of ascertaining the existence of Nash equilibria for some of these models as well as other models introduced here, as described below. 

Due to the complex nature of social events which might be woven with rational decisions, one can find various models aimed at capturing the idea of competition over social networks. One of the models that describes such a competitive behavior in networks is known as the competitive diffusion game \cite{alon2010note}. This model can be seen as a competition between two or more parties (types) who initially select a subset of nodes to invest, and the goal for each party is to attract as many social entities to his or her own type as possible. It was shown earlier \cite{alon2010note} that in general such games do not admit pure-strategy Nash equilibria. It has been shown in \cite{takehara2012comment} that such games may not even have a pure-strategy Nash equilibrium on graphs of small diameter. In fact, the authors in \cite{takehara2012comment} have shown that for graphs of diameter 2 and under some additional assumptions on the topology of the network, the diffusion game admits a general potential function and hence an equilibrium. However checking these assumptions at the outset for graphs of diameter at most 2 does not seem to be realistic. 

One of the advantages of the diffusion game model is that it captures the simple fact that being closer to player's initial seeds will result in adopting that specific player's type. Moreover, the adoption rule which is involved in the diffusion game is quite simple such that it enables each player to compute its best response quite fast with respect to others (at least for the case of single seed placement), given that all the other players have fixed their actions. On the other hand, as we will see in this paper what makes the analysis of such games more complicated is the behavior of nodes which are equally distanced from the players' seeds. Although there were some recent attempts to characterize these boundary points and show the existence of pure-strategy Nash equilibrium of the diffusion game over different types of networks \cite{small2013nash,small2012information,alon2010note}, in this work we will address this issue in a more general form, and show that finding an equilibrium for diffusion games is an NP-hard problem over general networks. Therefore, unless P=NP, this strongly suggests that in general the complexity of analyzing such a diffusion game is a hard task despite its simple adoption rule. It requires additional relaxations in the structure of the game in order to  make it more tractable. As one possible approach one may consider a probabilistic version of the diffusion game using some techniques from Markov chains or optimization of harmonic influence centrality \cite{acemoglu2013opinion}, \cite{vassio2014message}, \cite{acemoglu2014harmonic}. However, in this work we take a different approach by considering the diffusion game over the well-known Erdos-Renyi random graphs which are commonly used in the literature in order to model social networks.

In a related recent earlier work \cite{etesami2014complexity}, we have characterized the utilities of the players based on the graphical distances of various nodes from the initial seeds. In particular, we have studied the complexity of deciding on the existence of a pure-strategy Nash equilibrium. Here, we characterize the equilibria set for some classes of well-studied networks, and explore some connections between the set of equilibria and the underlying network structure. In particular, we provide some necessary conditions for a given profile of strategies to constitute a Nash equilibrium. Moreover, we consider the diffusion game over Erdos-Renyi graphs and prove some concentration results related to utilities of the players over such networks. Finally we provide a lower bound for the optimal social welfare of the diffusion game over general static networks based on their adjacency matrix.

The paper is organized as follows: in Section~\ref{sec:diffusion-game-model}, we describe the competitive diffusion game and review some of its properties and existing results regarding this model. In Section \ref{sec:specific-networks}, we determine the set of equilibria of two special but well-studied networks, namely the \textit{lattice} and the \textit{hypercube}. In Section~\ref{sec:NP-hardness}, we characterize the utilities of the agents based on the relative locations of the players' initial seed placements, and show that, the decision process on the existence of pure-strategy Nash equilibrium over general undirected networks is an NP-hard problem. In Section~\ref{sec:necessary}, we provide two necessary conditions based on the network structure for a given profile to be a Nash equilibrium.  In Section \ref{sec:diffusion-game-random-graphs}, we consider the diffusion game model over random graphs and show that asymptotically the utility of the players are highly concentrated around their mean. Furthermore, we provide a lower bound for the expected utility of the players based on the parameters of the random graphs. We end the paper with the concluding remarks of Section~\ref{sec:conclusion}. Finally, in Appendix \ref{ap-social-welfare-sub-modularity}, we provide some complementary results related to sub-modularity as well as lower optimal social welfare of the diffusion game over general fixed networks, which can be used to obtain bounds for the price of anarchy of diffusion games whenever an equilibrium exists.

\textbf{Notations and Conventions}: 
For a positive integer $n$, we let $[n]:=\{1,2,\ldots,n\}$. For a vector $v\in \R^n$, we let $v_i$ be the $i$th entry of $v$. Similarly, for a matrix $P$, we let $P_{ij}$ be the $ij$th entry of $P$ and we denote the $i$th row of $P$ by $P_i$. We denote the transpose of a matrix $P$ by $P'$. Moreover, we let $I$ and $\bold{1}$ be, respectively, the identity matrix and the column vector of all ones of proper dimensions. Given an integer $k>0$, we denote the set of all $k$-tuples of integers by $\mathbb{Z}^k$. For any two vectors $u,v\in Z^k$, we let $u\oplus v$ be their sum vector in mod 2, i.e., $(u\oplus v)_i=u_i+v_i \ \mbox{mod} 2$, for all $i=1,\ldots,k$. We let $\mathcal{G}=(V, \mathcal{E})$ to be an undirected graph with the set of vertices $V$ and the set of edges $\mathcal{E}$. We denote the degree of a vertex $x$ in graph $\mathcal{G}$ by $\mathbf{d}(x)$. Corresponding to $\mathcal{G}$ we let $\mathcal{A}_{\mathcal{G}}$ to be its adjacency matrix, i.e. $\mathcal{A}_{\mathcal{G}}(i,j)=1$ if and only if $(i,j)\in \mathcal{E}$ and $\mathcal{A}_{\mathcal{G}}(i,j)=0$, otherwise. Given a graph $\mathcal{G}=(V, \mathcal{E})$ and two vertices $x, y\in V$, we define $d_{\mathcal{G}}(x,y)$ to be the length of the shortest graphical path between $x$ and $y$. Also, for a set of vertices $S\subseteq V$ and a vertex $x$, we let $d_{\mathcal{G}}(x,S)=\min_{y\in S}\{d_{\mathcal{G}}(x,y)\}$. For a real number $a$ we let $\ceil{a}$ to be the smallest integer greater than or equal to $a$. We deal in this paper with only pure-strategy Nash equilibrium, and occasionally we will drop the qualifier ``pure-strategy".

\smallskip
\section{Competitive Diffusion Game}\label{sec:diffusion-game-model}

In this section we introduce the competitive diffusion game as was introduced earlier in \cite{alon2010note} and state some of the existing results for such a model. For sake of simplicity, and without much loss of conceptual generality, we state the model when there are only two players in the game; however, the model and analysis can readily be extended to the case when there are more than two players. 

\smallskip  
\subsection{Game Model} 
Following the formulation in \cite{alon2010note}, we consider here a network $\mathcal{G}$ of $n$ nodes and two players (types) $A$ and $B$. Initially at time $t=0$, each player decides to choose a subset of nodes in the network and place his own seeds. After that, a discrete time diffusion process unfolds among uninfected nodes as follows:
\begin{itemize}
\item If at some time step $t$ an uninfected node is neighbor to infected nodes of only one type ($A$ or $B$), it will adopt that type at the next time step $t+1$. 
\item If an uninfected node is connected to nodes of both types at some time step $t$, it will change to a gray node at the next time step $t+1$ and does not adopt any type afterward. 
\end{itemize}
This process continues until no new adoption happens. Finally, the utility of each player will be the total number of infected nodes of its own type at the end of the process. Moreover, if both players place their seeds on the same node, that node will change to gray. We want to emphasize the fact that when a node changes to gray, not only will it not adopt any type at the next time step, but also may block the streams of diffusion to other uninfected nodes. We will see later that the existence of gray nodes in the evolution of the process can make any prediction process about the outcome of the diffusion process much more complicated.  

\smallskip
\begin{remark} 
The diffusion process as defined above is a particular case of progressive diffusion processes, where the state of the nodes does not change after adoption. This is in contrast to other types of processes known as non progressive processes \cite{fazli2012non}.
\end{remark}   

\smallskip 
\begin{remark} 
For the case of $k>2$ players, one can define the same process as above such that an uninfected node will adopt type $i$ at time $t+1$ if and only if type $i$ is the only existing type among its neighbors at time step $t$.       
\end{remark}

\smallskip
As mentioned earlier, it has been shown in \cite{alon2010note}, \cite{small2012information} and \cite{takehara2012comment} that competitive diffusion game may or may not admit pure-strategy Nash equilibria depending on the topology of the network $\mathcal{G}$, and the number of the players. Moreover, it has been shown in \cite{small2013nash} that if the underlying graph $\mathcal{G}$ has a tree structure, the diffusion game with two players admits a pure-strategy Nash equilibrium. In fact, for the case of three or more players even the tree structure may not have a pure Nash equilibrium. In the next sections we will study some of the properties of the diffusion process over specific and general networks and in particular some conditions which are necessary for the existence of at least one equilibrium.  

\bigskip
\section{Diffusion Game over Specific Networks}\label{sec:specific-networks}

\smallskip
In this section, we consider the 2-player diffusion game with a single seed placement and study the existence of a pure-strategy Nash equilibrium for two special but well-studied classes of networks, namely the \textit{lattice} and the \textit{hypercube}. Such an analysis sheds light into the problem under more general settings, which is the topic of the next section.

\smallskip
\begin{definition}
An $m\times n$ lattice is a graph $L_{m\times n}$ with vertex set $V=\{(x,y)\in \mathbb{Z}^2: 0\leq x\leq m, 0\leq y\leq n\}$ such that each node is adjacent to those whose Euclidean distanced is 1 from it. A $k$ dimensional hypercube is a graph $Q_k$ with vertex set $\{0,1\}^k$ such that two $k$-tuples are adjacent if and only if they differ in exactly one position. 
\end{definition}

\smallskip
\begin{proposition}\label{prop:Lattice}
For the 2-player diffusion game over the lattice $L_{m\times n}, m,n\in Z^+$, a profile $(a^*,b^*)$ is a Nash equilibrium if and only if $a^*$ and $b^*$ are adjacent nodes in the most centric square or edges of the $L_{m\times n}$.   
\end{proposition}
\begin{proof}
Let us assume that $a^*=(x_1,y_1)$ and $b^*=(x_2,y_2)$ are two nonadjacent nodes in $L_{m\times n}$ which form a pure-strategy Nash equilibrium. Without any loss of generality, and if necessary by rotating the lattice $L_{m\times n}$ and relabeling the types, we can assume that $x_1>x_2$ and $y_1\ge y_2$. This situation has been shown in the Figure \ref{fig:Lattice}. Now, it is not hard to see that at least one of the following cases will happen:
\begin{itemize}
\item Player $B$ can strictly increase her utility by deviating to either $(x_1-1,y_1)$ or $(x_1,y_1-1)$.
\item Player $A$ can strictly increase her utility by deviating to either $(x_2+1,y_2)$ or $(x_2,y_2+1)$.
\end{itemize}
Therefore, in any of these cases, it means that $(a^*,b^*)$ cannot be an equilibrium. Moreover, if two adjacent nodes are not in the most centric part of the lattice, it can be seen that one of the players can always increase her utility by deviating to another neighbor of the other player while moves closer to the center of the lattice. This shows that if a profile is a Nash equilibrium it must be two adjacent nodes in the most centric square of the lattice $L_{m\times n}$. Finally, using the same line of argument as above, it is not hard to see that in every such profile each of the players will obtain the maximum utility that she can get, given that the position of the other player is fixed. This shows that indeed any adjacent profile in the most centric square is a Nash equilibrium. 
\vspace{-0.3cm}
\begin{figure}[htb] 
\vspace{-4cm}
\begin{center}
\includegraphics[totalheight=.3\textheight,
width=.55\textwidth,viewport=-120 0 500 550]{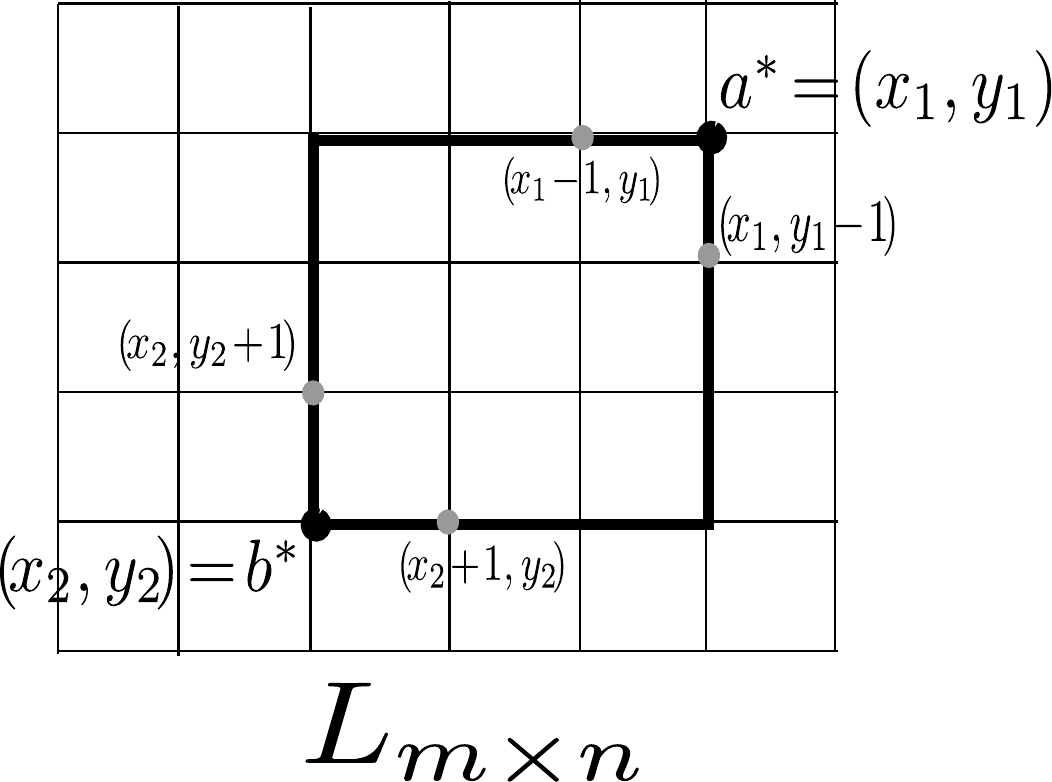}
\end{center}\vspace{-0.2cm}\caption{Illustration of Lattice in Proposition \ref{prop:Lattice}.}\label{fig:Lattice}
\end{figure}      
\end{proof}

\begin{proposition}\label{pro:hypercube}
A profile $(a^*,b^*)$ is a Nash equilibrium of the 2-player diffusion game over the hypercube $Q_k$ if and only if the graphical distance between $a^*$ and $b^*$ is an odd number, or equivalently $(a^*\oplus b^*)'\bold{1}_k=1\ \mbox{mod} \ 2$. 
\end{proposition}
\begin{proof}
First, we note that if $(a^*\oplus b^*)'\bold{1}_k=1\ \mbox{mod} \ 2$, then $U_A(a^*,b^*)=U_B(a^*,b^*)=2^{k-1}$, where $U_A(a^*,b^*)$ and $U_B(a^*,b^*)$ denote, respectively, the utilities of players $A$ and $B$ given that their initial seeds are at $a^*$ and $b^*$. In this case, there exists no vertex $v$ which has equal graphical distance to both $a^*$ and $b^*$, otherwise, if $d(a^*,v)=d(b^*,v)$, it means that $(a^*\oplus v)'\bold{1}=(b^*\oplus v)'\bold{1}$. Therefore, $(a^*\oplus b^*)'\bold{1}=(a^*\oplus v)'\bold{1}+(b^*\oplus v)'\bold{1}=0 \ \mbox{mod} 2$, which is in contradiction with the assumption. Therefore, using Lemma \ref{lemma:gray-eq}, every node of $Q_k$ must adopt either type $A$ or $B$. Finally, for every node $v$ such that $d(a^*,v)<d(b^*,v)$, one can assign a unique node $u=v\oplus (a^*\oplus b^*)$, such that $d(b^*,u)<d(a^*,u)$. This one-to-one bijection shows that $U_A(a^*,b^*)=U_B(a^*,b^*)=\frac{|V(Q_k)|}{2}=2^{k-1}$.   

To show that every such profile is indeed a pure Nash equilibrium, we show that for every arbitrary profile, the maximum utility that each player can gain is at most $2^{k-1}$. We show it by induction on $k$. For $k=1$, the result is trivial. Assuming that the statement is true for $k-1$, let $(a^*,b^*)$ be an arbitrary pair of nodes in $Q_k$. We consider two cases: 
\begin{itemize}
\item $a^*\oplus b^*=\bold{1}_k$. In this case $a^*$ and $b^*$ are binary complimentary of each other (by a simple translation of each node by the binary vector $a^*$ and without any loss of generality we may assume $a^*=0$ and $b^*=\bold{1}_k$). Therefore, by symmetry of $Q_k$, it is not hard to see that $U_A(a^*,b^*)=U_B(a^*,b^*)$. Since the total utility is at most $2^{k}$, thus we have $U_A(a^*,b^*)=U_B(a^*,b^*)\leq 2^{k-1}$. Such a situation for the case of $k=3$ is illustrated in Figure \ref{fig:hypercube}.   

\smallskip
\item $a^*\oplus b^*\neq \bold{1}_k$. In this case, there exists at least one coordinate such that $a^*$ and $b^*$ agree on it. Without any loss of generality, let us assume $a^*_i=b^*_i$ and let $Q=\{x\in V(Q_k): x_i=a^*_i\}$ and $Q^c=V(Q_k)\setminus Q$. Therefore, it is not hard to see that the induced subgraphs by vertices $Q$ and $Q^c$ are hypercubes of dimension $k-1$ which are connected to each other through a perfect matching (a set of disjoint edges which connect all the nodes in pairs). Furthermore, $a^*$ and $b^*$ are both located in $Q$. Therefore, by the induction hypothesis, the utility of the players on $Q$, i.e., $U_A^{Q}(a^*,b^*)$, $U_B^{Q}(a^*,b^*)$ is at most $2^{k-2}$. Moreover, since $Q$ and $Q^c$ are connected through a perfect matching and there is only one step delay in the diffusion process between $Q$ and $Q^c$, hence the adoption of vertices in $Q^c$ is exactly the same as those in $Q$. Therefore, we have $U_A(a^*,b^*)=2U_{A}^{Q}(a^*,b^*)\leq 2\times 2^{k-2}=2^{k-1}$, and similarly $U_B(a^*,b^*)\leq 2^{k-1}$.        
\end{itemize} 

Overall, we have shown that a profile $(a^*,b^*)$ in $Q_k$ is a Nash equilibrium if and only if $U_A(a^*,b^*)=U_B(a^*,b^*)=2^{k-1}$, which happens if and only if $a^*$ and $b^*$ have an odd graphical distance from each other, or equivalently $(a^*\oplus b^*)'\bold{1}_k=1\ \mbox{mod} \ 2$.

\begin{figure}[htb] 
\vspace{-4cm}
\begin{center}
\includegraphics[totalheight=.3\textheight,
width=.55\textwidth,viewport=-120 0 550 600]{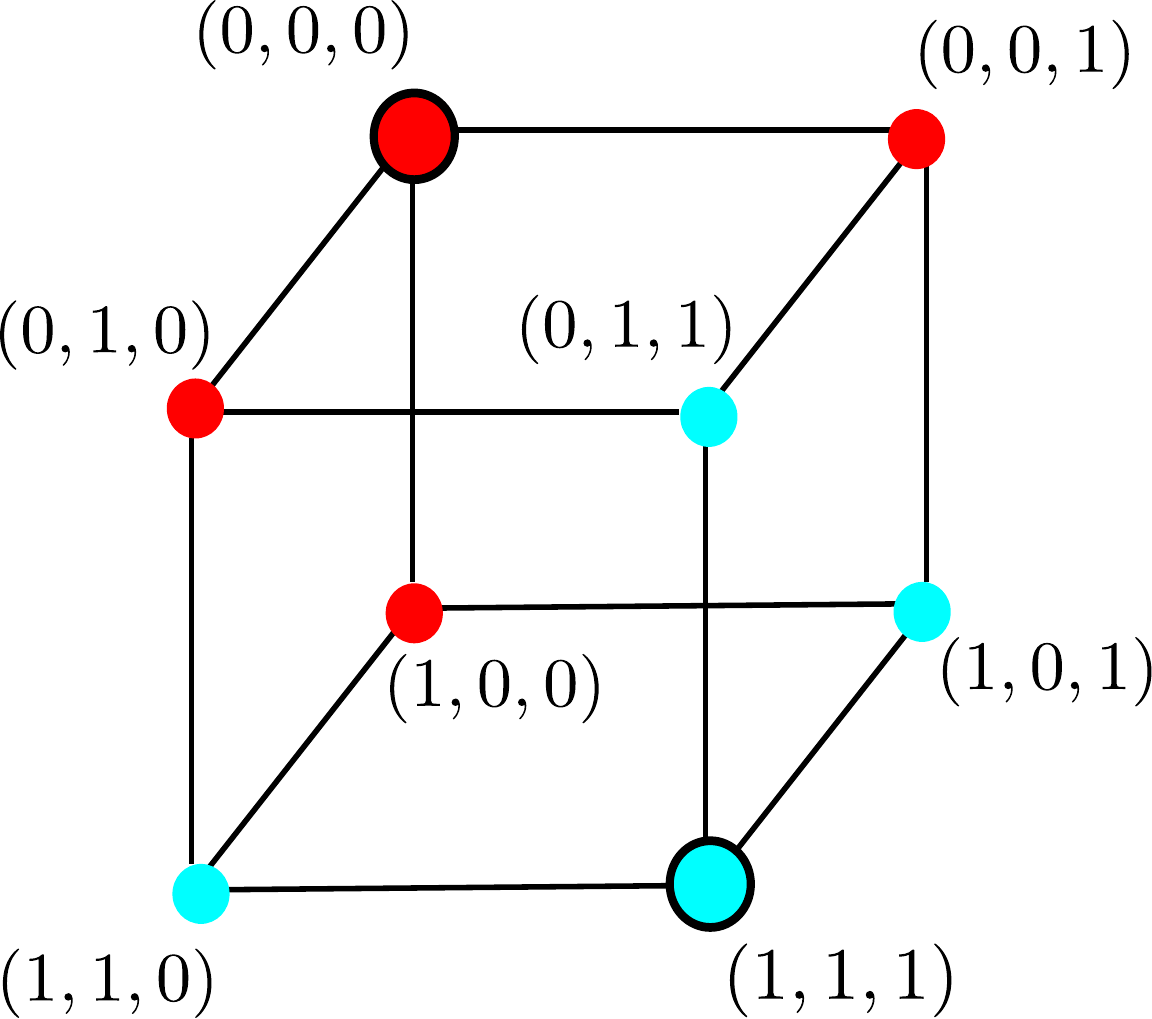}
\end{center}\vspace{-0.2cm}\caption{A Nash equilibrium of the diffusion game over 3-dimensional hypercube. The circled nodes denote the initial seeds and the figure illustrates the final state of the game.}\label{fig:hypercube}
\end{figure}          
\end{proof}

\bigskip
\section{NP-Hardness of Making a Decision on Existence of Nash Equilibrium}\label{sec:NP-hardness}

In this section, we first characterize the final state of the diffusion process based on relative distances of the players' initial seed placements on the network. Using this characterization we establish a hardness result for the decision process on the existence of a pure-strategy Nash equilibrium in the diffusion game.
 
\smallskip
\begin{lemma}\label{lemma:gray-eq}
Let $\mathcal{N}_A$ and $\mathcal{N}_B$ denote the set of nodes which adopt, respectively, types $A$ and $B$ at the end of the process for a particular initial selection of seeds $a, b\in V$. Moreover, let $\mathcal{N}$ be the set of gray or uninfected (white) nodes by the end of the process. Then, 
\begin{align}\nonumber
& \ \ \ \ \ \ \ \ \ \ \ \ \ \ \mathcal{N} \subseteq \{i: d_{\mathcal{G}}(a,i)=d_{\mathcal{G}}(b,i)\},\cr
&\{i: d_{\mathcal{G}}(a,i)\!<\! d_{\mathcal{G}}(b,i)\}\!\subseteq\!\mathcal{N}_A \!\subseteq\! \{i: d_{\mathcal{G}}(a,i)\!\leq\! d_{\mathcal{G}}(b,i)\}, \cr
&\{i: d_{\mathcal{G}}(b,i)\!<\! d_{\mathcal{G}}(a,i)\}\!\subseteq\! \mathcal{N}_B \!\subseteq\! \{i: d_{\mathcal{G}}(b,i)\!\leq\! d_{\mathcal{G}}(a,i)\}. \cr
\end{align}
\end{lemma}
\begin{proof}
Here we only sketch the steps of the proof and refer readers to \cite{etesami2014complexity} for the complete proof. The proof is by induction on the time step $t=0, 1,\ldots$. Assuming that the statement of the Lemma is true for all the gray nodes such that they are born at steps $t\leq k$, by considering different possibilities as to why a new node will change to gray at the next time step, one can show that this can happen only if the new node has the same graphical distance to the seed nodes. A similar argument shows that uninfected nodes can be reached from the seed nodes only by passing through some gray nodes, and hence, they must lie within the same graphical distance of the seed nodes. 
\end{proof}

\smallskip
Note that in Lemma \ref{lemma:gray-eq} we assumed that each player can only place one seed in the network as its initial placement. However, the result can be easily generalized to the case when each player (for example player $A$) is allowed to choose a set of nodes $S_A\subset V$ as its initial seed placements. In this case we just need to replace $d_{\mathcal{G}}(S_A,x)$ instead of $d_{\mathcal{G}}(a,x)$ in the statement of Lemma \ref{lemma:gray-eq} and all the other results carry naturally. Moreover, a similar result can be proved when there are more than two players in the game.

\smallskip
As it was shown before in \cite{alon2010note} and \cite{takehara2012comment}, one can always construct networks with diameter greater than or equal to 2 such that the diffusion game over such networks does not admit any pure-strategy Nash equilibrium. In fact, by a closer look at Lemma \ref{lemma:gray-eq}, one can see that there is some similarity between Voronoi games \cite{durr2007nash} and competitive diffusion games. Note, however, that in the competitive diffusion game a diffusion process unfolds while there is no notion of diffusion in Voronoi games. However, since at the end of the process both games demonstrate behavior close to each other, it seems natural to compare the complexity of Nash equilibria in these two games. In fact, in the following we show that the decision on the existence of Nash equilibrium in a diffusion game is NP-hard. Toward that goal, we first prove some relevant results and modify the configuration of the diffusion game to make a connection with Voronoi games. Borrowing some of the existing results from the Voronoi games literature, we prove the NP-hardness of verification of existence of Nash equlibrium in the diffusion game. We prove it by reduction from the 3-partitioning problem which is shown to be an NP-complete problem \cite{garey1975complexity}. In the 3-partitioning problem, we are given integers $\alpha_1, \alpha_2,\ldots, \alpha_{3m}$ and a $\beta$ such that $\frac{\beta}{4}<\alpha_i<\frac{\beta}{2}$ for every $1\leq i\leq 3m$, $\sum_{i=1}^{3m}\alpha_i=m\beta$ and have to partition them into disjoint sets $P_1, \ldots, P_m \subseteq \{1,2,\ldots,3m\}$ such that for every $1\leq j\leq m$ we have $\sum_{i\in P_j}\alpha_i=\beta$.
 
\smallskip
First, we briefly describe the stages that we will go through toward
proving the NP-hardness of making a decision on the existence of a Nash Equilibrium. Given an arbitrary network $\mathcal{G}$, by adding extra nodes and edges properly we expand this network into a new network $\bar{\mathcal{G}}$ such that we make sure that if there exists a Nash equilibrium in the original network, it must lie within a confined subset of nodes in the extended network $\bar{\mathcal{G}}$ (Lemma \ref{Lemma:U-magnified}). This allows us to confine our attention to only a specific subset of nodes in the extended network in order to search for a Nash equilibrium. Following this, we construct a new network (Theorem \ref{thm:NP-hard}) such that any Nash equilibrium of the game is equivalent to a solution of the 3-partitioning problem which is known to be an NP-complete problem. This establishes the NP-hardness of arriving at a decision on the existence of a Nash equilibrium in the diffusion game. We begin a formal proof of the result by stating the following lemma.    

\smallskip
\begin{lemma}\label{Lemma:U-magnified}
Let $T$ be a subset of $V(\mathcal{G})$. Then, there exists an extended graph $\bar{\mathcal{G}}$ such that there is a bijection between the Nash equilibria in $\mathcal{G}$ where the seeds (actions of the players) are restricted to $T$, and the unrestricted Nash equilibria in $\bar{\mathcal{G}}$. 
\end{lemma}
\begin{proof}
Consider the graph $\bar{\mathcal{G}}$ depicted in Figure \ref{fig:U-EQ}, which is constructed using $\mathcal{G}$ by adding $|T|n$ new nodes and $n\frac{|T|(|T|+1)}{2}$ new edges. Note that $|T|$ denotes the number of the nodes in $T$ and $n=2|V(\mathcal{G})|+1$ is a positive integer. It is important to note here that although $n>|V(\mathcal{G})|$, it is polynomially bounded above by $|V(\mathcal{G})|$. With each node $i\in T$ we associate a set of $n$ new nodes $C_i=\{c_{i1}, c_{i2}, \ldots, c_{in}\}$ and we connect all of them to node $i$. Furthermore, for $j=1,\ldots,n$, we connect all the nodes $c_{1j}, \ldots, c_{|T|j}$ to each other. In other words, nodes $\{c_{1j}, \ldots, c_{|T|j}\}$ form a clique for each $j=1,\ldots,n$. 

\begin{figure}[htb]
\begin{center}
\hspace{-2cm}\includegraphics[totalheight=.25\textheight,
width=.35\textwidth,viewport=0 0 700 700]{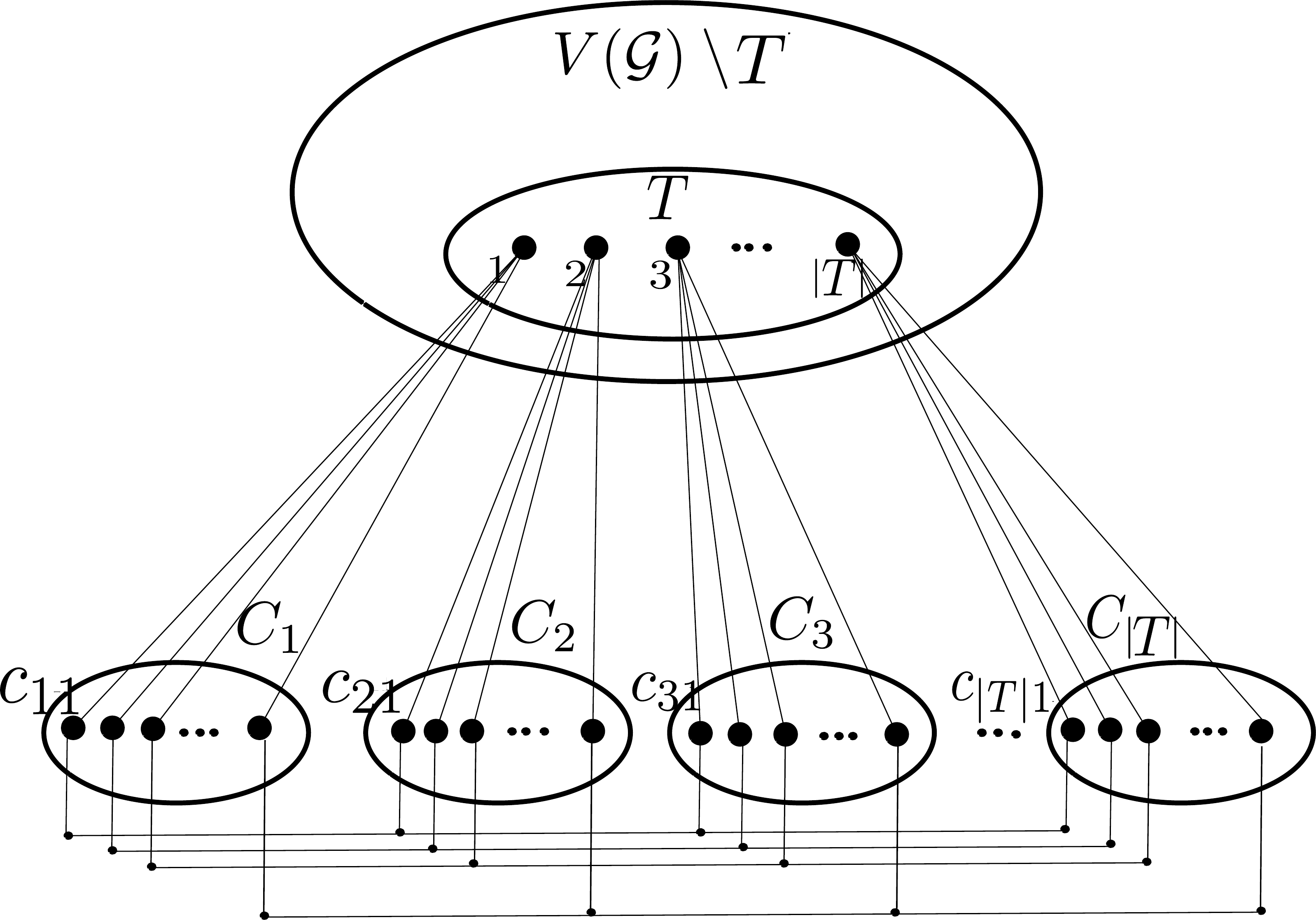}
\end{center}\caption{Extension of graph $\mathcal{G}$ to $\bar{\mathcal{G}}$}\label{fig:U-EQ}
\end{figure}

\smallskip
Now assume that at least one player puts his node seed in $k\in T$. We refer to this player as the \textit{first player} and denote its type by $A$. We claim that all the other players must play in $T$ as well. To prove this, suppose that another player which we refer to as the \textit{second player} with corresponding type $B$ chooses node $\ell\in T, \ell\neq k$. In this case he will earn at least $n$ due to winning all the nodes in $C_{\ell}$. Now let us assume that the second player plays in the bottom part of the graph (Figure \ref{fig:U-EQ}), i.e. $L=\cup_{i=1}^{|T|}C_i$. We consider two cases:
\begin{itemize}
\item[1.] He plays in $C_k$ and without any loss of generality and by symmetry, we may assume that he plays in $c_{k1}$. 
\item[2.] He plays in $C_{\ell}$ for some $\ell\neq k$ such as $c_{\ell 1}$.
\end{itemize}
 
\smallskip 
In the first case and after the first step of diffusion, all the elements $c_{11},c_{21},\ldots, c_{|T|1}$ will adopt type of the second player, i.e. $B$ (because there is a direct link between them and $c_{k1}$), and all the elements $c_{k2},c_{k3},\ldots, c_{kn}$ will adapt type of the first player, i.e. $A$. At the second time step, the second player can adopt all the elements of $T\setminus \{k\}$ in the best case. On the other hand, all the elements of $L\setminus \{c_{11},c_{21},\ldots, c_{|T|1}\}$ will change to $A$. Therefore, in this case the second player can gain at most $|T|+|V(\mathcal{G})|\leq 2|V(\mathcal{G})|<n$. 

\smallskip
In the second case and after the first step of diffusion, the second player can adopt only $\{c_{11}, c_{21}, \ldots, c_{|T|1},\ell\}\setminus \{c_{k1}\}$ while the first player will adopt at least $\{c_{k2},\ldots, c_{kn}\}$, (note that node $c_{k1}$ will change to gray). However, at the second time step all the nodes in $C_{\ell}\setminus \{c_{\ell 1}\}$ and also in $L\setminus \big\{C_k\cup C_{\ell}\cup \{c_{11}, c_{21}, \ldots, c_{|T|1}\}\big\}$ will change to type $A$. Therefore, in this case, the second player gains at most $|T|+|V(\mathcal{G})|\leq 2|V(\mathcal{G})|<n$.

\smallskip
Furthermore, if the second player places his seed at a node in $V(\mathcal{G})\setminus T$, then in the best scenario it will take at least two steps for the seed to be diffused to nodes of $L$. On the other hand, type $A$ can be diffused through every node of $L$ in no more than 2 steps. Thus all the nodes in $L$ either adopt $A$ or they change to gray and hence, in this case the second player can not earn more than $V(\mathcal{G})-1<n$. From the above discussion it should be clear that in either of the above cases, if a player is playing in $L \cup V(\mathcal{G})\setminus T$ he can always gain more by deviating to the set $T$. Thus in each equilibrium players must put their seeds in $T$.

\smallskip
Finally suppose that $\bar{Q}$ is a Nash equilibrium profile in $\bar{\mathcal{G}}$. By the above argument we know that all the players must play in $T$ and thus, each of these players gains exactly $n$ from the set $L$. Therefore, the utility of players is equal to the utility that they would gain by playing on $\mathcal{G}$ plus $n$. This shows that $\bar{Q}$ must be an equilibrium for $\mathcal{G}$ when the strategies of players are restricted to $T$. Similarly, if $Q$ is an equilibrium of $\mathcal{G}$ when the strategies of players are restricted to $T$, it is also an equilibrium for $\bar{\mathcal{G}}$ as we know all the equilibria seeds of players (if there is any) must be in $T$. This shows that the set of equilibria of $\bar{\mathcal{G}}$ is equivalent to the set of equilibria of $\mathcal{G}$ with the restricted strategy set $T$.                  
\end{proof}

\bigskip
\begin{theorem}\label{thm:NP-hard}
Given a graph $\mathcal{G}$ and $m\ge 2$ players, the decision process on the existence of pure-strategy Nash equilibrium for the diffusion game on $\mathcal{G}$ is NP-hard. 
\end{theorem}
\smallskip
\begin{proof}
Assume that integers $\alpha_1, \alpha_2,\ldots, \alpha_{3m}$ and $\beta>3$ are given such that $\frac{\beta}{4}<\alpha_i<\frac{\beta}{2}$ for every $1\leq i\leq 3m$, $\sum_{i=1}^{3m}\alpha_i=m\beta$. Moreover, let $c={3m \choose 3}$ and choose an integer $d$ such that $\frac{(\beta-1)c}{4}<d<\frac{\beta c}{4}$. Consider the graph depicted in Figure \ref{fig:Complexity-graph}, where the set $T$ is defined to be the set of vertices whose induced subgraph $\mathcal{G}[T]$ is illustrated in the dashed-line area. In fact, such a graph is composed of three parts:
\begin{itemize}
\item The right graph: This is a graph of $9d$ nodes, composed of 9 stars of size $d$ where the centers of the stars are connected as it is shown in Figure \ref{fig:Complexity-graph}. It has been shown in \cite{takehara2012comment} that this graph does not admit any pure-strategy Nash equilibrium when there are two agents. In fact, if there is only one player on this graph, he will gain $9d$ and if there are two players on it, one of them can always deviate to gain at least $4d$.
\item The middle graph: This graph is simply a clique of size ${3m \choose 3}$ where the nodes are labeled by $u_{i,j,k}$ for all possible triples $\{i,j,k\},\  i,j,k\in[3m]$.
\item The left graph: This graph is composed of $3m$ independent sets (i.e., there exists no edge between vertices of each set) $\mathcal{I}_1,\ldots,\mathcal{I}_{3m}$ of sizes $c\alpha_1,\ldots,c\alpha_{3m}$, respectively, such that all the nodes in the independent sets $\mathcal{I}_i,\mathcal{I}_j$, and $\mathcal{I}_{k}$ are connected to node $u_{i,j,k}$ in the middle graph.
\end{itemize}

\vspace{0.5cm}   
\begin{figure}[htb]
\vspace{1cm} 
\begin{center}
\hspace{-3.8cm}
\includegraphics[totalheight=.22\textheight,
width=.31\textwidth,viewport=50 0 750 700]{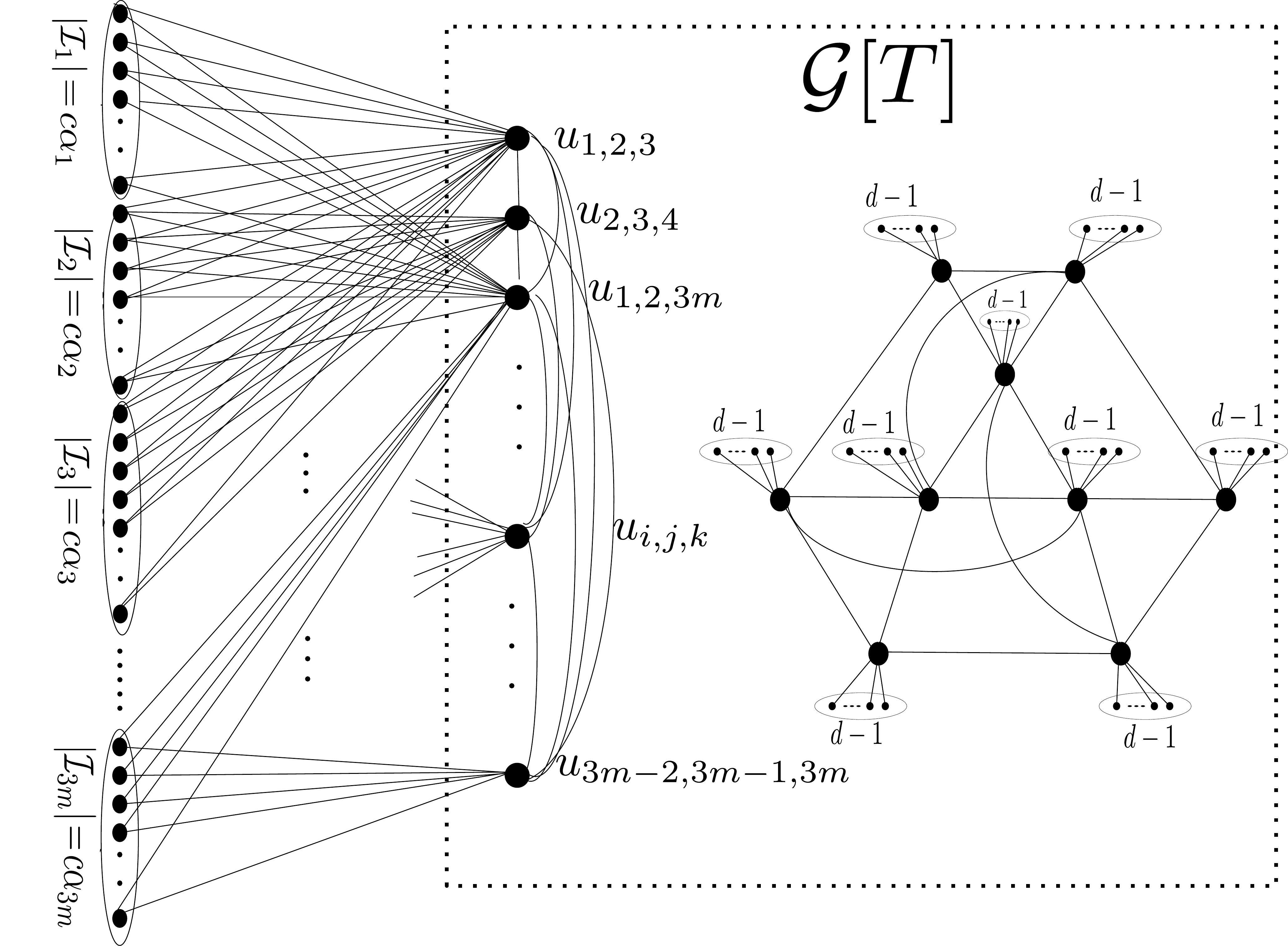}
\end{center}\caption{Graph construction of Theorem \ref{thm:NP-hard}}\label{fig:Complexity-graph}
\end{figure}

Calling the graph in Figure \ref{fig:Complexity-graph} $\mathcal{G}$, setting $T$ to be the set of vertices of the middle and right side graphs, and applying Lemma \ref{Lemma:U-magnified} to construct $\bar{\mathcal{G}}$, we can see that any Nash equilibrium of $\bar{\mathcal{G}}$ is an equilibrium of $\mathcal{G}$ when the strategies of players are restricted to $T$. We claim that $Q$ is an equilibrium for $\mathcal{G}$ with the restricted strategy set $T$ (and equivalently an equilibrium for unrestricted $\bar{\mathcal{G}}$) if and only if there is a 3-partitioning of $\{\alpha_i\}_{i=1}^{3m}$. 

Assume that there is a solution $P_1, \ldots, P_m$ to the 3-partition. For every $1\leq q\leq m$, if $P_q=\{i,j,k\}$, then player $q$ is assigned to $u_{i,j,k}$. Let also assume that player $m+1$ is assigned to one of the nodes in the most right part of the graph, which makes his utility to be $9d$. If player $m+1$ moves to a vertex $u_{i,j,k}$ his utility will be $1<9d$, because all the other $m$ players already covered all the $\sum_{\ell=1}^{3m}ca_{\ell}$ nodes in the most left side of the graph and his movement will not result in any additional payoff for him except producing some gray nodes. Now if one of the players $1\leq q\leq m$ moves from vertex $u_{i,j,k}$ to one of the nodes in the right part of the graph, then his gain can be at most $4d$ which by the selection of $d$ would be less than what he was getting before ($\beta c$). Finally, If player $q$ or equivalently node $u_{i,j,k}$ moves to another node $u_{i',j',k'}$ for some $\{i', j', k'\}\neq \{i, j, k\}$ then since $P_q$ was part of 3-partitioning before, it means that his payoff after deviating will be at most $c\max\{\alpha_i+\alpha_j, \alpha_i+\alpha_k, \alpha_j+\alpha_k\}<c\beta$. Moreover, by Lemma \ref{Lemma:U-magnified} no player at equilibrium will be out of $T$ and hence the proposed profile using the 3-partitioning is an equilibrium.

Now let us suppose that there exists a Nash equilibrium for $\bar{\mathcal{G}}$. We show that it corresponds to a solution of 3-partitioning. First we note that there cannot be two players in the most right part of the graph, otherwise it is not an equilibrium by \cite{takehara2012comment}. Moreover, if there are 3 players or more, one of them can gain at most $3d$. Since in this case there are at most $m-2$ players in the middle part, we can find a set $\{i',j',k'\}$ such that the corresponding set of all the other players does not have any intersection with it. Therefore, if a player with the least gain (at most $3d$) from the right side deviates to $u_{i',j',k'}$ in the middle part, he will gain at least $\frac{3\beta c}{4}$ which is greater than $3d$. Thus the most right part can have either one player or nothing. However since $9d>\frac{3\beta c}{2}$, at least one player would want to move to the most right part of the graph if there is no other player there. Therefore, the most right part of the graph has exactly one player. Thus, the rest of the $m$ players must not only play in $T$ (because of Lemma \ref{Lemma:U-magnified}) but also they must form a partition. Otherwise one of them can move to an appropriate vertex of the middle graph and increase his utility. Finally, in this partitioning, each player must gain exactly $\beta c$, because if this is not true and one of the players namely $u_{i,j,k}$ gets less than $\beta c$, he will gain at most $(\beta-1)c$ (note that $c$ is a rescaling factor) and thus, he can always move to the most right side of the graph and gain $4d>(\beta-1)c$. Thus this partitioning must be a 3-partitioning. This proves the equivalence of the existence of Nash equilibrium in $\bar{\mathcal{G}}$ and existence of 3-partitioning for the set $\{\alpha_i\}_{i=1}^{3m}$.          
\end{proof}

\smallskip
\section{Necessary conditions for Nash equilibrium in Competitive Diffusion Game}\label{sec:necessary}

In this section, we consider the 2-player diffusion game with a single seed placement, present some necessary conditions for existence of pure-strategy Nash equilibrium, and discuss its connection with the network structure. Here, it is worth noting that although the results of this section provide some necessary conditions for a given profile to be a Nash equilibrium, in general they do not guarantee the existence of a Nash equilibrium, that is they are not sufficient. We start with the following theorem which is for the case of two players; however, it can be extended quite naturally to the case of an arbitrary number of players.

\smallskip
\begin{theorem}\label{thm:EQ-bound}
Suppose that $(a^*,b^*)\in V \times V$ is an equilibrium profile for the diffusion game. Then, 
\begin{align}\nonumber
\ceil{\frac{n\!-\!1}{\mathbf{d}(a^*)}}\leq U_B(a^*,b^*), \ \  \ceil{\frac{n\!-\!1}{\mathbf{d}(b^*)}}\leq U_A(a^*,b^*),
\end{align}
where $U_A(a^*,b^*)$ and $U_B(a^*,b^*)$ denote the utilities of players $A$ and $B$ given the initial seed placement at $(a^*,b^*)$, and $\mathbf{d}(a^*)$ and $\mathbf{d}(b^*)$ denote, respectively, the degrees of nodes $a^*$ and $b^*$ in the network $\mathcal{G}$. 
\end{theorem}
\begin{proof}
Assume that player $A$ and $B$ place their seeds at nodes $a^*$ and $b^*$ and receive payoffs of $U_A(a^*,b^*)$ and $U_B(a^*,b^*)$, respectively. We claim that there exists a neighbor of $a^*$ where player $b^*$ can gain at least $\ceil{\frac{n-1}{\mathbf{d}(a^*)}}$ by deviating to it. Toward showing this, let us also denote all the neighbors of $a^*$ by $v_1,v_2,\ldots, v_{\mathbf{d}(a^*)}$. Let us denote the nodes that adopt $B$ for the initial seed allocation $(a^*,v_i)$ by $S_{i}, i=1,2,\ldots,\mathbf{d}(a^*)$. Then, we have $\cup_{i=1}^{\mathbf{d}(a^*)}S_i=V\setminus \{a^*\}$. In fact, for every $v\in V\setminus \{a^*\}$, the shortest path from $v$ to $a^*$ must pass through at least one of the neighbors of $a^*$ such as $v_{\ell}$. This means that $d_{\mathcal{G}}(v,v_{\ell})<d_{\mathcal{G}}(v,a^*)$ and using Lemma \ref{lemma:gray-eq} we can see that $v\in S_{\ell}$. Therefore we have $n-1=|\cup_{i=1}^{\mathbf{d}(a^*)}S_i|\leq \sum_{i=1}^{\mathbf{d}(a^*)}|S_{i}|$ and this means that there exists at least one $v_{i^*}$ such that $|S_{i^*}|\ge\ceil{\frac{|V\setminus \{a^*\}|}{\mathbf{d}(a^*)}}=\ceil{\frac{n-1}{\mathbf{d}(a^*)}}$. Since we assumed that $(a^*,b^*)$ is an equilibrium, player $b^*$ can not gain more by deviating to $v_{i^*}$. This means that $\ceil{\frac{n\!-\!1}{\mathbf{d}(a^*)}}\leq U_B(a^*,b^*)$ and using the same argument for the other player we get $\ceil{\frac{n\!-\!1}{\mathbf{d}(b)}}\leq U_A(a^*,b^*)$.
\end{proof}

\smallskip
Note that the results in Theorem \ref{thm:EQ-bound} can be improved by noting that the inequality $|\cup_{i=1}^{\mathbf{d}(a^*)}S_i|\leq \sum_{i=1}^{\mathbf{d}(a^*)}|S_{i}|$ can be strict and there are nodes which might be counted in different sets of $S_i$. In fact, it is not hard to see that if a node $v$ belongs to two of these sets such as $S_{j}$ and $S_{k}$, $v$ must be in an even cycle emanating $a^*$ and including the nodes $v_j$ and $v_k$. In such a case, to every cycle of even length which includes $a^*$ and does not contain another smaller cycle, one can associate a node which is counted twice in two different sets. We call such cycles \textit{simple even cycles emanating from $a^*$}. Therefore, we can write
\begin{align}\nonumber
n-1&=|\cup_{i=1}^{\mathbf{d}(a^*)}S_i|\cr 
&\leq \sum_{i=1}^{\mathbf{d}(a^*)}|S_{i}|\!-\!\{\mbox{Simple even cycles emanating from $a^*$}\}|,
\end{align} 
and therefore the bound in the Theorem \ref{thm:EQ-bound} will change to 

\begin{small}
\begin{align}\nonumber
\Big\lceil{\frac{n\!-\!1\!+\!|\{\mbox{Simple even cycles emanating from $a^*$}\}|}{\mathbf{d}(a^*)}\Big\rceil}\!\leq\! U_B(a^*,b^*).  
\end{align}
\end{small}

Next, we consider the following two definitions from graph theory.
\smallskip 
\begin{definition}
An edge (a vertex) of a connected graph $\mathcal{G}$ is a \textit{cut-edge (cut-vertex)} if its removal disconnects the graph. 
\end{definition}

\smallskip
\begin{definition}
A \textit{block} of a graph $\mathcal{G}$ is a maximal connected subgraph of $\mathcal{G}$ that has no cut-vertex.
\end{definition} 

\smallskip
\begin{remark}\label{rem:block-diagram}
Two blocks in a graph share at most one vertex. Hence the blocks of a graph partition its edge set. Furthermore, a vertex shared by two blocks must be a cut-vertex.
\end{remark}  

\bigskip
\begin{theorem}\label{thm:block}
Every pure-strategy Nash equilibrium of a 2-player diffusion game must lie within one of the blocks of its underlying network. 
\end{theorem}
\begin{proof}
Given a network $\mathcal{G}=(V,\mathcal{E})$ with block decomposition $\mathcal{B}_1, \mathcal{B}_2,\ldots, \mathcal{B}_k$, let us denote one of the pure-strategy Nash equilibria of the diffusion game on $\mathcal{G}$ by $(a^*, b^*)\in V\times V$. If $a^*, b^* \in \mathcal{B}_i$, for some $i\in [k]$, then there is nothing to prove. Otherwise, without any loss of generality, let us assume that $a^*\in \mathcal{B}_1$ and $b^*\in \mathcal{B}_2$. Starting from node $a^*$ and moving along a shortest path $\mathcal{P}$ between $a^*$ and $b^*$, the path must exit block $\mathcal{B}_1$ for the first time at some vertex $v_1$. Clearly by Remark \ref{rem:block-diagram} such a vertex must be a cut vertex since it is shared between two blocks. By a similar argument, but this time by starting from node $b^*$ and moving along the path $\mathcal{P}$ we can see that the path $\mathcal{P}$ must exit block $\mathcal{B}_2$ through a cut vertex $v_2$. Next we consider two cases:
\begin{itemize}
\item $v_1=v_2:=v$. In this case, we can assume that neither $a^*=v$ nor $b^*=v$, otherwise $(a^*,b^*)$ would already be an equilibrium within one block (either $\mathcal{B}_1$ or $\mathcal{B}_2$). Now, without any loss of generality let us assume that the vertex $v$ does not adopt type $A$ for the seed placement $(a^*, b^*)$. Then, the first player can strictly increase its utility by removing its seed from $a^*$ and placing it in node $v$. In fact, placing the initial seeds on $(v,b^*)$ instead of $(a^*,b^*)$, the first player not only can adopt all the nodes that he was able to adopt in $(a^*, b^*)$ (this is due to the fact that there is no edge between vertices of two distinct blocks (Remark \ref{rem:block-diagram})), but also he can adopt at least one more new node, which is node $v$. This is in contrast with $(a^*,b^*)$ being a Nash equilibrium. Hence $(a^*, b^*)$ must lie within one block. 
\item $v_1\neq v_2$. In this case, by the definition of the cut vertices $v_1$ and $v_2$ we clearly have $d_{\mathcal{G}}(a^*,v_1)<d_{\mathcal{G}}(a^*,v_2)$, and $d_{\mathcal{G}}(b^*,v_2)<d_{\mathcal{G}}(b^*,v_1)$ (Figure \ref{fig:block-diagram}). Now, if $v_2$ adopts type $A$, by Lemma \ref{lemma:gray-eq} it is straightforward to see that $v_1$ adopts type $A$ as well. But then by the same argument as in the first case, the second player can move its seed from $b^*$ to $v_1$ and strictly increase its utility. On the other hand, if $v_2$ does not adopt type $A$, then the first player can move its seed from $a^*$ to $v_2$, and strictly increase its utility. This is in contradiction with $(a^*,b^*)$ being a Nash equilibrium, and hence $(a^*, b^*)$ must lie within one block.
\end{itemize}
\end{proof}

\begin{figure}[htb]
\vspace{-3.5cm} 
\begin{center}
\hspace{-3cm}\includegraphics[totalheight=.25\textheight,
width=.29\textwidth,viewport=0 0 800 800]{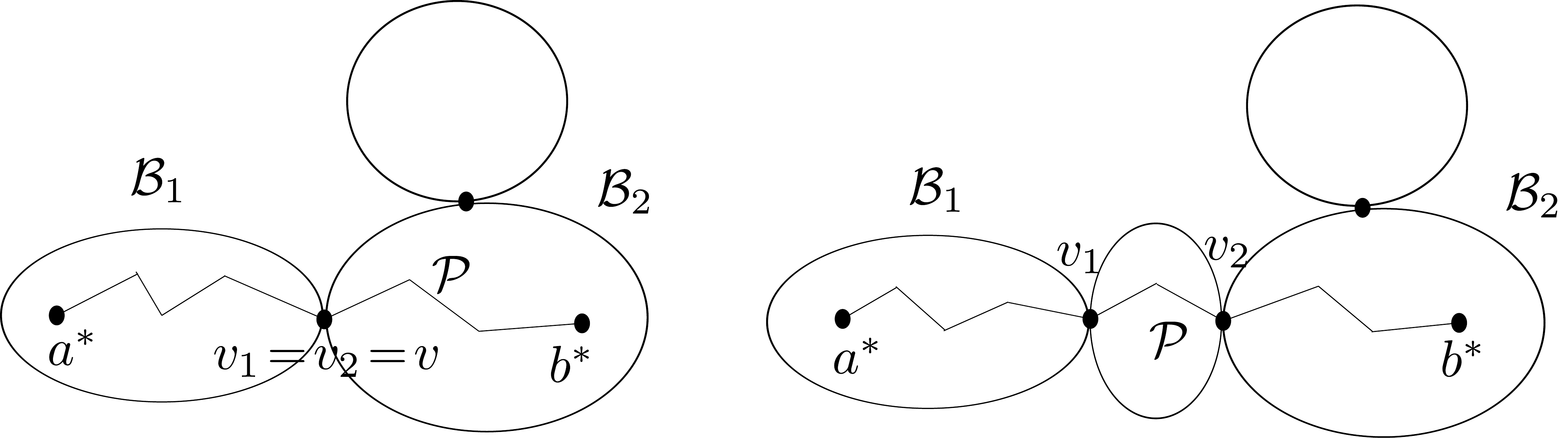}
\end{center}\caption{Illustration of block diagrams in Theorem \ref{thm:block}.}\label{fig:block-diagram}
\end{figure}

We note that Theorem \ref{thm:block} breaks down the complication of search of a Nash equilibrium over the entire graph to only its blocks. We further note that there are efficient time algorithms which can decompose a graph into its blocks in at most $O(n^2)$ steps, where $n$ is the number of the nodes in the graph \cite{paton1971algorithm}.

\section{Diffusion Game over Random Graphs}\label{sec:diffusion-game-random-graphs}

Social networks that are observed in the real world, can be viewed as a single realization of an underlying stochastic process. This line of thinking has generated a huge interest in modeling real world social networks using random networks, which are highly structured while being essentially random. This is in fact a property which emerges in many real world social networks. One of the most well-known random structures in this context is the Erdos-Renyi graph $G(n,p)$, where there are $n$ nodes and the edges emerge independently with probability $p\in (0,1)$.         

In this section we consider a two player diffusion game with single seed placement over the Erdos-Renyi graph $\mathcal{G}(n,p)$. It is a well-known fact \cite{bollobas1998random} that $p(n)=\frac{\ln n}{n}$ is a threshold function for the connectivity of the random graph $G(n,p)$. In particular, for $p\ge \frac{\ln n}{n}$, almost surely there does not exist any isolated vertex in $G(n,p)$, as $n\to \infty$. On the other hand, it was shown earlier \cite{bollobas1998random,alon2010note}, that $p(n)=\sqrt{\frac{2\ln n}{n}}$ is a threshold function for having diameter 2 in $G(n,p)$. In particular, for $p(n)\ge \sqrt{\frac{c\ln n}{n}}, c>2$ almost all the nodes in $G(n,p)$ lie withing graphical distance of at most 2 from each other, which results in some straight forward analysis of the diffusion game over such graphs. Therefore, in this section we confine our attention to the more interesting region where $p\in (\frac{\ln n}{n}, \sqrt{\frac{2\ln n}{n}})$. 

\smallskip
For any arbitrary but fixed node $x$ and any realization $\mathcal{G}$ of $G(n,p)$, we let $S_{\mathcal{G}}(i)$ and $B_{\mathcal{G}}(i)$ be the set of all the nodes which are, respectively,  at graphical distances of exactly $i$, and at most $i$ from node $x$. Similarly, we define $S(i)$ and $B(i)$ be two random sets denoting, respectively, the set of nodes of distances exactly $i$, and at most $i$ from node $x$ when the underlying graph is a random graph $G(n,p)$. Now we have the following lemma. 

\smallskip
\begin{lemma}\label{lemm:S-B-relation}
For an arbitrary $\lambda>(n-1)p$, and any $i\in [n]$ we have 
\begin{align}\nonumber
\mathbb{P}(|S(i)|\ge \lambda |B(i-1)|)\leq ne^{-\frac{(\lambda-(n-1)p)^2}{3(n-1)p}}.
\end{align} 
\end{lemma}
\begin{proof}
\begin{align}\nonumber
&\mathbb{P}\Big(|S(i)|\ge \lambda|B(i-1)|\Big)=\mathbb{P}\left(\frac{|S(i)|}{|B(i-1)|}\ge\lambda\right)\cr 
&\qquad\leq \mathbb{P}\Big(\exists v\in B(i-1): \mathbf{d}(v)\ge \lambda\Big)\leq \mathbb{P}\Big(\exists v: \mathbf{d}(v)\ge \lambda\Big)\cr &\qquad\leq n\mathbb{P}(\mathbf{d}(v)\!\ge\! \lambda)=n\mathbb{P}\Big(\mathbf{d}(u)\!-\!(n\!-\!1)p\ge \lambda\!-\!(n\!-\!1)p\Big)\cr 
&\qquad\leq ne^{-\frac{(\lambda-(n-1)p)^2}{3(n-1)p}},
\end{align}
where in the first inequality we have used the fact that every node in $S(i)$ must be connected to at least one vertex in $B(i-1)$. Therefore, there exists at least one vertex in $B(i-1)$ whose degree is greater than or equal to $\frac{|S(i)|}{|B(i-1)|}$. Finally, in the last inequality we have used the Chernoff bound for the random variable $\mathbf{d}(v)$. Note that $\mathbf{d}(v)$ is the sum of $(n-1)$ independent Bernoulli random variables with equal probability of occurrence $p$.     
\end{proof} 

\smallskip
Next we consider the following definition.

\begin{definition}
A martingale is a sequence $X_0, X_1,\ldots, X_m$ of random variables, where for $0\leq i< m$, we have $\mathbb{E}[X_{i+1}|X_i,X_{i-1},\ldots,X_0]=X_i$. It is easy to see that if $X_0, X_1,\ldots, X_m$ is a martingale, then $\mu:=X_0=\mathbb{E}[X_i], \forall i=1,2\ldots,m.$     
\end{definition}

\smallskip
\begin{lemma}\label{lemm:Azuma}[Azuma's Inequality \cite{alon2004probabilistic}]
Let $\mu=X_0, X_1,\ldots, X_m$ be a martingale with $|X_{i+1}-X_i|\leq 1$ for all $0\leq i< m$. Then 
\begin{align}\nonumber
\mathbb{P}\Big(|X_m-\mu|>\sqrt{m}\theta\Big)\leq 2e^{-\frac{\theta^2}{2}}.
\end{align}
\end{lemma}

\smallskip
Now we are ready to state the main result of this section.

\smallskip
\begin{theorem}
For any arbitrary constant $\alpha\in(0,1)$, let $p\in[\frac{\ln n}{n},\sqrt{\frac{\ln n}{n^{1+\alpha}}}]$. Then, as $n\to\infty$, for every random seed placement in the 2-player diffusion game over $G(n,p)$ with single seed, we have $\mathbb{E}[U_A]=\mathbb{E}[U_B]\ge \frac{1}{5p}$. Moreover, as $n\to\infty$ almost surely we have $\frac{U_A}{\mathbb{E}[U_A]}=\frac{U_B}{\mathbb{E}[U_B]}\to 1$.   
\end{theorem}
\begin{proof}
For an arbitrary but fixed node $x$, let $I(x)$ be the event that in the random graph $G(n,p)$ with uniform seed placements at $(a,b)\in V\times V$, node $x$ gets infected (adopts either of the two types), and we denote its complement by $I^c(x)$. Conditioning on the graph $G(n,p)=\mathcal{G}$, we can write
\begin{align}\label{eq:two-random-inequalities}
&\mathbb{P}(I(x)|G(n,p)=\mathcal{G})\ge \frac{\sum_{i\neq j}|S_{\mathcal{G}}(i)||S_{\mathcal{G}}(j)|}{n^2}\cr 
&\mathbb{P}(I^c(x)|G(n,p)=\mathcal{G})\leq \frac{\sum_{i=1}^{n}|S_{\mathcal{G}}(i)|^2}{n^2}, 
\end{align}
where the first inequality is due to the fact that if the seed nodes lie in different sets $S_{\mathcal{G}}(i), S_{\mathcal{G}}(i), i\neq j$ from node $x$, then by Lemma \ref{lemma:gray-eq} node $x$ will adopt one of the types. Moreover, for the second inequality and using Lemma \ref{lemma:gray-eq}, one can see that the set of uninfected or gray nodes is a subset of vertices which are within equal distance from the seed nodes. 

Now let us define the following sets of graphs;
\begin{align}\nonumber
&\mathcal{K}=\{\mathcal{G}:\mathbf{d}_{\mathcal{G}}(v)>0, \forall v\}\cr 
&\mathcal{F}=\{\mathcal{G}: |S_{\mathcal{G}}(i)|<\lambda |B_{\mathcal{G}}(i-1)|, \forall i\in [n]\}.
\end{align}
By combining the two inequalities in \eqref{eq:two-random-inequalities} and taking the average over the probability space of all possible graphs with distribution $G(n,p)$, we can write

\begin{align}\label{eq:random-I-vs-I^c}
\frac{\mathbb{P}(I^c(x))}{\mathbb{P}(I(x))}&=\frac{\sum_{\mathcal{G}}\mathbb{P}(\mathcal{G})\mathbb{P}(I^c(x)|G(n,p)=\mathcal{G})}{\sum_{\mathcal{G}}\mathbb{P}(\mathcal{G})\mathbb{P}(I(x)|G(n,p)=\mathcal{G})}\cr 
&\leq\frac{\sum_{\mathcal{G}}\mathbb{P}(\mathcal{G})\left(\sum_{i=1}^{n}|S_{\mathcal{G}}(i)|^2\right)}{\sum_{\mathcal{G}}\mathbb{P}(\mathcal{G})\left(\sum_{i\neq j}|S_{\mathcal{G}}(i)||S_{\mathcal{G}}(j)|\right)}\cr 
&=\frac{\sum\limits_{\mathcal{G}\in\mathcal{F}}\mathbb{P}(\mathcal{G})\!\left(\sum\limits_{i=1}^{n}|S_{\mathcal{G}}(i)|^2\right)\!\!+\!\!\!\!\sum\limits_{\mathcal{G}\in\mathcal{F}^c}\mathbb{P}(\mathcal{G})\!\left(\sum\limits_{i=1}^{n}|S_{\mathcal{G}}(i)|^2\right)}{\sum\limits_{\mathcal{G}}\mathbb{P}(\mathcal{G})\!\left(\sum\limits_{i\neq j}|S_{\mathcal{G}}(i)||S_{\mathcal{G}}(j)|\right)}\cr 
&\leq \frac{\sum\limits_{\mathcal{G}\in\mathcal{F}}\mathbb{P}(\mathcal{G})\!\left(\sum\limits_{i=1}^{n}|S_{\mathcal{G}}(i)|^2\right)+n^2\mathbb{P}\Big(\mathcal{G}\in \mathcal{F}^c\Big)}{\sum\limits_{\mathcal{G}\in\mathcal{F} }\mathbb{P}(\mathcal{G})\!\left(\sum\limits_{i\neq j}|S_{\mathcal{G}}(i)||S_{\mathcal{G}}(j)|\right)},
\end{align}
where the second inequality holds because $\sum_{i=1}^{n}|S_{\mathcal{G}}(i)|^2\leq n^2$. On the other hand, we have 
\begin{align}\label{eq:sphere-to-ball}
\sum\limits_{i\neq j}|S_{\mathcal{G}}(i)||S_{\mathcal{G}}(j)|&=2 \sum\limits_{i=1}^{n}|S_{\mathcal{G}}(i)|\sum\limits_{j\leq i-1}|S_{\mathcal{G}}(j)|\cr 
&=2\sum\limits_{i=1}^{n}|S_{\mathcal{G}}(i)||B_{\mathcal{G}}(i-1)|.
\end{align} 

Using \eqref{eq:sphere-to-ball} in \eqref{eq:random-I-vs-I^c} we can write.

\begin{align}\label{eq:random-second-I-I^c}
&\frac{\mathbb{P}(I^c(x))}{\mathbb{P}(I(x))}\leq \frac{\sum\limits_{\mathcal{G}\in\mathcal{F}}\mathbb{P}(\mathcal{G})\!\left(\sum\limits_{i=1}^{n}|S_{\mathcal{G}}(i)|^2\right)+n^2\mathbb{P}\Big(\mathcal{G}\in \mathcal{F}^c\Big)}{2\sum\limits_{\mathcal{G}\in\mathcal{F} }\mathbb{P}(\mathcal{G})\!\sum\limits_{i=1}^{n}|S_{\mathcal{G}}(i)||B_{\mathcal{G}}(i-1)|}\cr 
&\qquad\leq \frac{\sum\limits_{\mathcal{G}\in\mathcal{F}}\mathbb{P}(\mathcal{G})\!\left(\sum\limits_{i=1}^{n}|S_{\mathcal{G}}(i)|^2\right)+n^2\mathbb{P}\Big(\mathcal{G}\in \mathcal{F}^c\Big)}{\frac{2}{\lambda}\sum\limits_{\mathcal{G}\in\mathcal{F} }\mathbb{P}(\mathcal{G})\!\left(\sum\limits_{i=1}^{n}|S_{\mathcal{G}}(i)|^2\right)}\cr
&\qquad=\frac{\lambda}{2}+\frac{n^2\lambda}{2}\frac{\mathbb{P}\Big(\mathcal{G}\in \mathcal{F}^c\Big)}{\sum\limits_{\mathcal{G}\in\mathcal{F} }\mathbb{P}(\mathcal{G})\!\left(\sum\limits_{i=1}^{n}|S_{\mathcal{G}}(i)|^2\right)}\cr 
&\qquad \leq\frac{\lambda}{2}+\frac{n^2\lambda}{2}\frac{\mathbb{P}\Big(\mathcal{G}\in \mathcal{F}^c\Big)}{\sum\limits_{\mathcal{G}\in\mathcal{F}\cap \mathcal{K} }\mathbb{P}(\mathcal{G})\!\left(\sum\limits_{i=1}^{n}|S_{\mathcal{G}}(i)|^2\right)}\cr 
&\qquad \leq\frac{\lambda}{2}+\frac{n^2\lambda}{2}\frac{\mathbb{P}\Big(\mathcal{G}\in \mathcal{F}^c\Big)}{\mathbb{P}\Big(\mathcal{G}\in \mathcal{F}\cap\mathcal{K}\Big)}, 
\end{align}
where in the second inequality we have used the definition of set $\mathcal{F}$, and in the last inequality we have used the definition of $\mathcal{K}$ to get $\sum\limits_{i=1}^{n}|S_{\mathcal{G}}(i)|^2\ge 1, \forall \mathcal{G}\in \mathcal{F}\cap\mathcal{K}$. On the other hand, by Lemma \ref{lemm:S-B-relation} and using the union bound, we can write,
\begin{align}\label{eq:probability-of-F}
\mathbb{P}(\mathcal{G}\in \mathcal{F}^c)&=\mathbb{P}(\exists i\in [n]: |S(i)|\ge \lambda |B(i-1)|)\cr 
&\leq \sum_{i}^{n}\mathbb{P}(|S(i)|\ge \lambda |B(i-1)|)\cr 
&\leq n^2e^{-\frac{(\lambda-(n-1)p)^2}{3(n-1)p}}.
\end{align} 

By choosing $\lambda=(1+\sqrt{15})np$ for sufficiently large $n$, and using \eqref{eq:probability-of-F}, each of the probabilities $\mathbb{P}(\mathcal{G}\in \mathcal{F})$ and $\mathbb{P}(\mathcal{G}\in \mathcal{K})$ can be made arbitrarily close to 1. Therefore, for sufficiently large $n$, we have $\mathbb{P}(\mathcal{G}\in\mathcal{F}\cap \mathcal{K})\ge \frac{1}{2}$. Hence, substituting \eqref{eq:probability-of-F} in \eqref{eq:random-second-I-I^c} and using $\mathbb{P}(\mathcal{G}\in\mathcal{F}\cap \mathcal{K})\ge \frac{1}{2}$, we get 
\begin{align}\nonumber
\frac{\mathbb{P}(I^c(x))}{\mathbb{P}(I(x))}&\leq  \frac{\lambda}{2}+ \lambda n^4e^{-\frac{(\lambda-(n-1)p)^2}{3(n-1)p}}\cr 
&=\frac{1\!+\!\sqrt{15}}{2}np\!+\!(1\!+\!\sqrt{15})n^5pe^{-\frac{((1+\sqrt{15})np-(n-1)p)^2}{3(n-1)p}}\cr 
&\leq \frac{1+\sqrt{15}}{2}np+(1+\sqrt{15})n^5p e^{-\frac{15\ln n}{3}}\cr 
&= \frac{1+\sqrt{15}}{2}np+(1+\sqrt{15})p, 
\end{align}
where the second inequality holds because $p\ge \frac{\ln n}{n}$. Finally, since $\mathbb{P}(I^c(x))=1-\mathbb{P}(I(x))$, we get $\mathbb{P}(I(x))\ge \frac{2}{2+(2+2\sqrt{15})p+(1+\sqrt{15})np}$. Now because of the symmetry between players we can write, 
\begin{align}\nonumber
\mathbb{E}[U_A]&=\mathbb{E}[U_B]=\frac{1}{2}\mathbb{E}[U_A+U_B]=\frac{1}{2}\sum_x P(I(x))\cr 
&\ge \frac{n}{2+(2+2\sqrt{15})p+(1+\sqrt{15})np}.
\end{align} 
Therefore, for sufficiently large $n$, we have $\mathbb{E}[U_A]\ge \frac{1}{5p}$. 

\smallskip
Now for two arbitrary but fixed nodes $a$ and $b$, let us assume that players $A$ and $B$ place their seeds at $a$ and $b$, respectively. We introduce a martingale by setting $X_0=\mathbb{E}[U_A(a,b)]$, and $X_i=\mathbb{E}[U_A(a,b)|v_1=a, v_2=b, v_3,v_4,\ldots,v_i]$ to be the expected utility of player $A$ after the first $i$ nodes of the random graph $G(n,p)$ are exposed. In other words, to find $X_i$ we expose the first $i$ vertices and all their internal edges and take the conditional expectation of $U_A$ with that partial information. It is straight forward to check that this defines a martingale of length at most $n$ such that $|X_{i+1}-X_i|\leq 1$, as adding one vertex can at most change the utility of player $A$ by at most 1. Therefore, using Azumas' inequality (Lemma \ref{lemm:Azuma}) we can write
\begin{align}\nonumber
\mathbb{P}\Big(|U_A-\mathbb{E}[U_A]|\ge \sqrt{n}\theta\Big)\leq 2e^{-\frac{\theta^2}{2}}.
\end{align} 
In particular, since $\mathbb{E}[U_A]\ge \frac{1}{5p}$, we can write
\begin{align}\nonumber
&\mathbb{P}\Big(\left|\frac{U_A}{\mathbb{E}[U_A]}-1\right|\ge 5p\sqrt{n}\theta\Big)\cr 
&\qquad\leq\mathbb{P}\Big(\left|\frac{U_A}{\mathbb{E}[U_A]}-1\right|\ge\frac{\sqrt{n}\theta}{\mathbb{E}[U_A]}\Big)\leq 2e^{-\frac{\theta^2}{2}}. 
\end{align}
By choosing $\theta=\ln n$ and since $p\in(\frac{\ln n}{n}, \sqrt{\frac{\ln n}{n^{1+\alpha}}})$, one can see that for any $\epsilon>0$, there exists a sufficiently large $n(\epsilon)$ such that for $n\ge n(\epsilon)$, we have $\mathbb{P}\Big(\left|\frac{U_A}{\mathbb{E}[U_A]}-1\right|\ge\epsilon\Big)\leq \epsilon$. By the same argument and by symmetry, we can see that for sufficiently large $n$, $U_B$ is arbitrarily close to its mean. Since $\mathbb{E}[U_B]=\mathbb{E}[U_A]$, we can write,
\begin{align}\nonumber
&\mathbb{P}\left(\left|\frac{U_A}{\mathbb{E}[U_A]}\!-\!1\right|\!<\!\epsilon,\ \left|\frac{U_B}{\mathbb{E}[U_B]}\!-\!1\right|\!<\!\epsilon \right)\cr &\qquad=\!1\!-\!\mathbb{P}\left(\left|\frac{U_A}{\mathbb{E}[U_A]}\!-\!1\right|\!\ge\!\epsilon, \mbox{or} \left|\frac{U_B}{\mathbb{E}[U_B]}\!-\!1\right|\!\ge\!\epsilon\right)\cr 
&\qquad\ge 1-\mathbb{P}\left(\left|\frac{U_A}{\mathbb{E}[U_A]}\!-\!1\right|\!\ge\!\epsilon\right)\!-\!\mathbb{P}\left(\left|\frac{U_B}{\mathbb{E}[U_B]}\!-\!1\right|\!\ge\!\epsilon\right)\cr 
&\qquad\ge 1-2\epsilon.
\end{align} 
This completes the proof.
\end{proof}   

\smallskip
As we close this section, we stress the fact that the subset of nodes
which adopt either of the two types in diffusion games can be viewed as a community whose
members have closer interactions with each other. In other words, in the diffusion game
each community can be viewed as the final subset of nodes that adopt a specific technology,
which raises the question of efficient decomposition of the network into different communities.
In such problems, the main issue is to partition the set of nodes into different groups such
that the set of edges within each group is much larger than that between the groups. Different approaches in order to determine the communities effectively so that they scale with the parameters of the network under both deterministic and randomized setting have been proposed in the literature \cite{hajek2014computational}, \cite{newman2004detecting}. As an example, an electrical voltage-based approach in order to determine the communities within a network such that they scale linearly with the size of the network has been discussed in \cite{Huberman2003communities}.

\bigskip
\section{Conclusion}\label{sec:conclusion}
In this paper, we have studied a class of games known as diffusion games which model the competitive behavior of a set of social actors on an undirected connected social network. We determined the set of pure-strategy Nash equilibria for two special but well-studied classes of networks. We showed
that, in general, making a decision on the existence of Nash equilibrium for such a class of games is an NP-hard problem. Further, we have presented some necessary conditions for a given profile to be an equilibrium in general graphs. Finally, we have studied the behavior of the competitive diffusion game over Erdos-Renyi graphs, obtained some concentration results, and derived lower bounds for the expected utilities of the players over such random structures.      

As a future direction of research, an interesting problem is identifying the class of networks which admit Nash equilibria for the case of two players in the competitive diffusion game. It is not hard to see that a tree construction which is a special case of bipartite graphs always leads to a pure-strategy equilibrium \cite{small2013nash} for the case of two players. Finally, studying a more robust model of the diffusion games with respect to changes in the network topology is another important problem. 

%%%%%%%%%%%%%%%%%%%%%%%%%%%%%%%%%%%%%%%%%%%%%%%%%%%%%%%%%%%%%%%%%%%%%%%%%%%%%%%%%%%%%%%%%%%%%%%%%%%%%%%%%%%%    

\bigskip       
\textbf{\textit{Acknowledgement}}:
We would like to thank Esther Galbrun for sharing with us a counterexample regarding the sub-modular property of utility functions in the diffusion game problem (Example \ref{ex:counter-example}). In the earlier version of this work, listed as \cite{etesami2014complexity}, we had erroneously argued that the players' utilities in the diffusion game model is a sub-modular function of their initial seed placements. Moreover, we would like to thank Vangelis Markakis and Vasileios Tzoumas for bringing to our attention their relevant work on the competitive diffusion game \cite{tzoumas2012game}. 

\bibliographystyle{IEEEtran}
\bibliography{thesisrefs}

%%%%%%%%%%%%%%%%%%%%%%%%%%%%%%%%%%%%%%%%%%%%%%%%%%%%%%%%%%%%%%%%%%%%%%%%%%%%%%%%%%%%%%%%%%%%%%%%%%%%%%%%%%%
\appendices
\smallskip
\section{Social Welfare and sub-modularity in the Competitive Diffusion Game}\label{ap-social-welfare-sub-modularity}

\smallskip
In this appendix we first study the maximum social welfare of the game which can be achieved by both players in the competitive diffusion game with a single seed placement. The motivation for such a study is that any bound which is obtained for the optimal social welfare can be used to provide some bound on the price of anarchy of the game, which measures the degradation in the efficiency of a system due to selfish behavior of its agents, and is defined to be the ratio of the centralized optimal social welfare over the sum utilities of the worst equilibrium. Following that, we study the sub-modular property in the competitive diffusion game. It was shown in the literature that greedy algorithm optimization approaches work quite well for the dynamics which benefit from sub-modular property \cite{mossel2010submodularity}. In this appendix we show that unlike some other diffusion processes \cite{kempe2003maximizing}, the utilities of the players in the competitive diffusion game do not have the sub-modular property.

We start with the following definition:

\smallskip
\begin{definition}
Given an arbitrary nonnegative matrix $P$, we define its zero pattern $\sigma(P)$ to be
\begin{align}\nonumber
\sigma(P)_{ij}=\begin{cases} 1, & P_{ij}>0 \\
                           0, & P_{ij}=0.
\end{cases}
\end{align}  
\end{definition} 

\smallskip
\begin{theorem}\label{thm:upper-social-welfare}
Given a graph $\mathcal{G}=(V,\mathcal{E})$ of $n$ nodes and diameter $D$, and two players $A$ and $B$, there exists an initial seed placement $(a, b)$ for players such that the social utility $U_A(a,b)+U_B(a,b)$ is at least 
\begin{align}\nonumber
n+1-\frac{\sum_{k=1}^{D}\big\|\big(\sigma(\mathcal{A}^k)-\sigma(\mathcal{A}^{k-1})\big)\bold{1}\big\|^2}{n(n-1)}, 
\end{align} 
where $\|\cdot\|$ denotes the standard Euclidean norm, and $\mathcal{A}=I+\mathcal{A}_{\mathcal{G}}$ where $\mathcal{A}_{\mathcal{G}}$ is the adjacency matrix of the network $\mathcal{G}$.   
\end{theorem} 
\begin{proof}
Let us define $\mathcal{G}^{(k)}=(V, \mathcal{E}^{(k)})$, where $\mathcal{E}^{(k)}=\{(i,j)| d_{\mathcal{G}}(i,j)=k\}$.
We consider all the initial placements over different pairs of nodes, and then compute the average utility gained by players. To do that, we consider an array $Q$ of $n \choose 2$ rows and $n$ different columns. For $i\neq j, k\in \{1,2,\ldots n\}$, we let $Q(\{i,j\},k)=1$ if and only if node $k$ adopts either $A$ or $B$ during the diffusion process for the initial placement $\{i, j\}$, and $Q(\{i,j\},k)=0$, otherwise (Figure \ref{fig:matrix}). We count the maximum number of zeros in $Q$. For an arbitrary but fixed node $x\in V=\{1,2,\ldots,n\}$, we count the number of different initial seed placements which result in node $x$ turning to gray. In other words, we count the maximum number of zeros in column $x$ of $Q$. For this purpose we note that, using Lemma \ref{lemma:gray-eq}, if node $x$ turns to gray, it must be at equal distance from seed nodes. On the other hand, for a given $k=1, 2, \ldots, D$, the number of choosing two nodes at distance $k$ from node $x$ (as possible seed placements which may turn $x$ to gray) is the same as the number of choosing two nodes among neighbors of $x$ in $\mathcal{G}^{(k)}$, i.e. $d_{\mathcal{G}^{(k)}}(x) \choose 2$, where $d_{\mathcal{G}^{(k)}}(x)$ denotes the degree of node $x$ in $\mathcal{G}^{(k)}$. Thus, the maximum number of zeros in column $x$ of $Q$ is upper bounded by $\sum_{k=1}^{D}{d_{\mathcal{G}^{(k)}}(x) \choose 2}$ and hence the number of zeros in $Q$ is bounded from above by $\sum_{x \in V}\sum_{k=1}^{D}{d_{\mathcal{G}^{(k)}}(x) \choose 2}$. This means that the average number of ones in each row of $Q$ is at least 

\begin{figure}[htb]
\vspace{-4.8cm}
\begin{center}
\includegraphics[totalheight=.35\textheight,
width=.45\textwidth,viewport=-90 0 310 400]{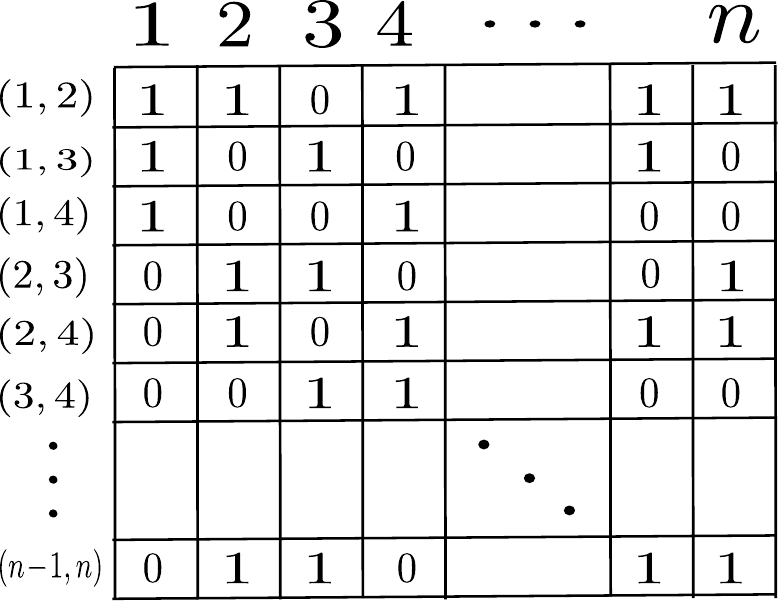}
\end{center}\vspace{-0.2cm}\caption{Illustration of the array $Q$ in Theorem \ref{thm:upper-social-welfare}.}\label{fig:matrix}
\end{figure}
\vspace{-0.3cm}
\begin{align}\label{eq:Social-welfare-bound}
&\frac{n{n \choose 2}\!-\!\sum_{x \in V}\sum_{k=1}^{D}\!{d_{\mathcal{G}^{(k)}}(x) \choose 2}}{{n \choose 2}}\!=\!n\!-\!\frac{\sum_{k=1}^{D}\sum_{x \in V}\!{d_{\mathcal{G}^{(k)}}(x) \choose 2}}{{n \choose 2}}\cr&\qquad=n+\frac{\sum_{k=1}^{D}\sum_{x \in V}d_{\mathcal{G}^{(k)}}(x)}{n(n-1)}\!-\!\frac{\sum_{k=1}^{D}\sum_{x \in V}{d_{\mathcal{G}^{(k)}}^2(x)}}{n(n-1)}\cr &\qquad=n+\frac{\sum_{k=1}^{D}2|\mathcal{E}^{(k)}|}{n(n-1)}-\frac{\sum_{k=1}^{D}\sum_{x \in V}{d_{\mathcal{G}^{(k)}}^2(x)}}{n(n-1)}\cr &\qquad=n+1-\frac{\sum_{k=1}^{D}\sum_{x \in V}{d_{\mathcal{G}^{(k)}}^2(x)}}{n(n-1)},
\end{align}
where the last equality follows because $\{\mathcal{E}^{(k)}\}_{k=1}^{D}$ partitions all the edges of a complete graph with $n$ nodes. This yields a lower bound on the maximum social welfare for the case of two players on the graph. 

Finally, using the zero pattern definition of a matrix, we can compute the last quantity in \eqref{eq:Social-welfare-bound} in the following way. 
Let us take $A_{\mathcal{G}}$ to be the adjacency matrix of graph $\mathcal{G}$ of diameter $D$ and $\mathcal{A}=I+\mathcal{A}_{\mathcal{G}}$, where $I$ denotes the identity matrix of appropriate dimension. It is not hard to see that $\sigma(\mathcal{A}^k)-\sigma(\mathcal{A}^{k-1})$ is the adjacency matrix of $\mathcal{G}^{(k)}$. In other words, $\mathcal{A}_{\mathcal{G}^{(k)}}=\sigma(\mathcal{A}^k)-\sigma(\mathcal{A}^{k-1})$. Therefore, if we let $\bold{1}$ be the column vector of all ones, the degree of each node $x$ in $\mathcal{G}^{(k)}$ can be found easily by looking at the $x$ coordinate of vector $\big[\sigma(\mathcal{A}^k)-\sigma(\mathcal{A}^{k-1})\big]\bold{1}$. Thus we can write
\begin{align}\nonumber
&\sum_{x\in V}\sum_{k=1}^{D}d_{\mathcal{G}^{(k)}}^2(x)\cr&\qquad=\sum_{k=1}^{D}\bold{1}'\big[\sigma(\mathcal{A}^k)-\sigma(\mathcal{A}^{k-1})\big]'\big[\sigma(\mathcal{A}^k)-\sigma(\mathcal{A}^{k-1})\big]\bold{1}\cr 
&\qquad=\sum_{k=1}^{D}\big\|\left(\sigma(\mathcal{A}^k)-\sigma(\mathcal{A}^{k-1})\right)\bold{1}\big\|^2. 
\end{align}         
\end{proof}

Note that the total number of operations needed to compute the expression given in Theorem \ref{thm:upper-social-welfare} is at most polynomial in terms of the number of nodes in the graph. 

\smallskip
\begin{definition}
Given a set $\Omega$, a set function $f:2^{\Omega}\rightarrow \mathbb{R}$ is called a sub-modular function if for any two subsets $S,\bar{S}\subseteq \Omega$ with $S\subseteq \bar{S}$ and any $x\in \Omega\setminus \bar{S}$, we have 
\begin{align}\nonumber
f(S\cup \{x\})-f(S)\ge f(\bar{S}\cup \{x\})-f(\bar{S}).
\end{align}
\end{definition}

\smallskip
In fact, the following example which is due to Esther Galbrun shows that the sub-modular property does not hold for the diffusion game model introduced in Section \ref{sec:diffusion-game-model}.
\smallskip
\begin{example}\label{ex:counter-example}[\textit{A Counterexample by Esther Galbrun}]
In this example, there are two players red (r) and blue (b). Circled nodes are initial seeds and the graph shows the color of the nodes at the end of the diffusion. In all cases blue only picks $v_3$. Moreover, $S=\emptyset\subseteq \{v_1\}=\bar{S}$ and $x=\{v_4\}$. As it can be seen in Figure \ref{fig:counter-example-sub-modular}, the utility function of the red player does not satisfy sub-modular property.

\begin{figure}[htb]
\vspace{-2cm}
\begin{center}
\includegraphics[totalheight=.25\textheight,
width=.35\textwidth,viewport=100 0 800 700]{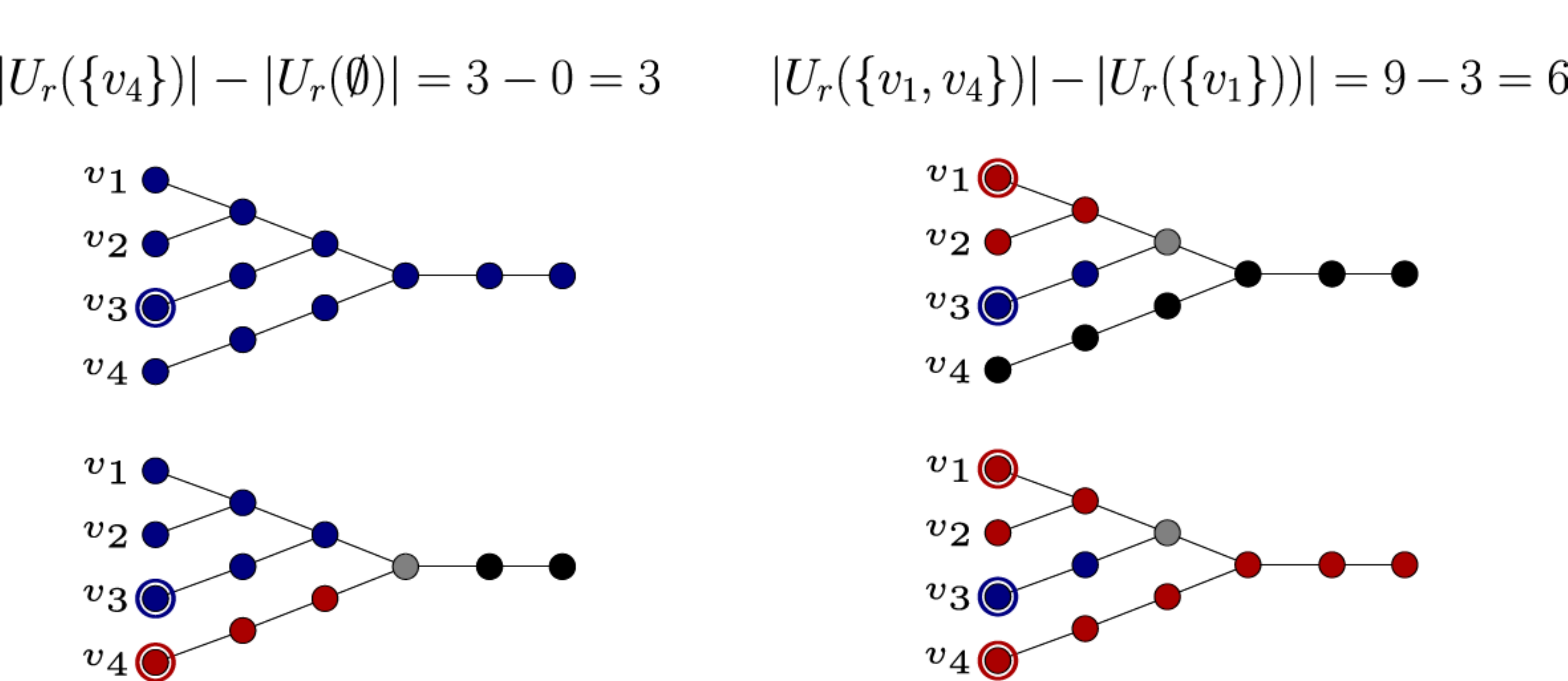}
\label{Fig:U-EQ}
\end{center}\caption{A counterexample for the existence of sub-modular property in the diffusion game}\label{fig:counter-example-sub-modular}
\end{figure}
\end{example}

\end{document}